%% file: main.tex
\documentclass[11pt]{article}
\usepackage[T1]{fontenc}
\usepackage{lmodern}
\usepackage[margin=1in]{geometry}
\usepackage{microtype}
\usepackage{amsmath,amssymb,amsthm,mathtools}
\usepackage{xcolor}
\input{macros}

\usepackage{tikz}
\usetikzlibrary{arrows.meta,calc,positioning}
\renewcommand{\OPT}{\text{\normalfont\scshape Opt}}
\renewcommand{\ALG}{\text{\normalfont\scshape Alg}}
\definecolor{ink}{HTML}{263940}
\definecolor{teal}{HTML}{237D87}
\definecolor{ochre}{HTML}{B87628}
\definecolor{pathcolor}{HTML}{7253A0}
\tikzset{
  vertex/.style={circle,draw=ink,fill=white,inner sep=0pt,minimum size=4.5pt},
  subtreevertex/.style={vertex,draw=teal,fill=teal},
  boundaryvertex/.style={rectangle,draw=ochre,fill=white,inner sep=0pt,minimum size=5.5pt},
  treeedge/.style={draw=ink!55,line width=.55pt},
  patharrow/.style={draw=pathcolor,line width=1.2pt,-{Stealth[length=4pt]}}
}
\title{An $\tilde \Omega(\log n \log m)$ Information-Theoretic
Lower Bound for Randomized Online Set Cover}
\author{Roie Levin\thanks{Department of Computer Science, Rutgers University, New Jersey. E-mail: \texttt{roie.levin@rutgers.edu}.}}
\date{}
\begin{document}

\maketitle

\begin{abstract}
	We show an information-theoretic lower bound of $\Omega\left(\frac{\log n \log m}{\log \log n + \log \log m}\right)$ for online set cover against randomized algorithms, for all sufficiently large $m$ and $n$ satisfying $\log^2 n \leq m \leq 2^n$.
\end{abstract}
\input{intro.tex}
\input{overview.tex}
\input{inner.tex}
\input{amplify.tex}

\input{general_parameters.tex}

\section{AI Disclosure}

This result was obtained via extensive interaction with both ChatGPT 6 Astra and Claude Fable 5.1. The author prompted both models by describing the inner/outer instance framework and setting as the goal to find an inner instance with parameters $n = 2^{\Theta(\sqrt{m})}$ and competitive ratio $\Omega(\sqrt{m})$. Astra discovered the inner instance essentially independently. The author takes full responsibility for any mathematical errors.

\clearpage
{\footnotesize
\bibliography{refs}
\bibliographystyle{alpha}
}
\end{document}

%% file: macros.tex
\usepackage{typearea}
\usepackage{setspace}

\usepackage{fullpage}

\usepackage{comment} 
\usepackage{bbm}
\usepackage{framed} 
\usepackage{url} 
\usepackage{complexity}
\usepackage{booktabs}
\usepackage{amsmath,amssymb}
\usepackage{float}

\usepackage{amsthm}
\usepackage{thmtools} 
\usepackage{thm-restate}
\usepackage{nicefrac}
\usepackage{calc}
\usepackage{enumerate}
\usepackage{enumitem}

\usepackage[ruled,vlined,linesnumbered]{algorithm2e}
\usepackage{forloop}

\usepackage{graphicx}
\usepackage[font=footnotesize,labelfont=bf]{subcaption}
\usepackage[font=footnotesize,labelfont=bf]{caption}
\usepackage[nobreak=true]{mdframed}
\usepackage{appendix}
\usepackage[noend]{algpseudocode}
\usepackage[colorinlistoftodos]{todonotes}
\usepackage{xr}
\usepackage{array}
\usepackage{xspace}
\definecolor{ForestGreen}{rgb}{0.1333,0.5451,0.1333}
\definecolor{DarkRed}{rgb}{0.8,0,0}
\definecolor{Red}{rgb}{1,0,0}
\usepackage[linktocpage=true,
pagebackref=true,colorlinks,
linkcolor=DarkRed,citecolor=ForestGreen,
bookmarks,bookmarksopen,bookmarksnumbered]{hyperref}

\usepackage[capitalise]{cleveref}
\crefrangelabelformat{enumi}{#3#1#4--#5#2#6}
\usepackage{parskip}
\usepackage{chngcntr}
\usepackage{mathtools,stackengine}
\usepackage{multirow}

\stackMath
\newcommand{\stackGeq}[1]{%
	\setbox0=\hbox{${}\mathrel{\stackon[-1pt]{\geq}{\scriptstyle\text{#1\strut}}}{}$}
	\xdef\tmpwd{\dimexpr\the\wd0\relax}
	\kern.5\tmpwd\mathclap{\box0}&\kern.5\tmpwd
}

\usepackage{nameref}
\usepackage{tikz}
\usetikzlibrary{patterns}

\allowdisplaybreaks

\DeclareMathOperator*{\expectation}{\mathbb{E}}
\let\poly\relax
\DeclareMathOperator*{\poly}{poly}

\newcommand\eps{\epsilon}

\newcommand\abs[1]{\lvert #1 \rvert}
\newcommand{\expect}{\expectation\expectarg}
\DeclarePairedDelimiterX{\expectarg}[1]{[}{]}{%
	\ifnum\currentgrouptype=16 \else\begingroup\fi
	\activatebar#1
	\ifnum\currentgrouptype=16 \else\endgroup\fi
}

\DeclarePairedDelimiterX{\nicesetarg}[1]{\{}{\}}{%
	\ifnum\currentgrouptype=16 \else\begingroup\fi
	\activatebar#1
	\ifnum\currentgrouptype=16 \else\endgroup\fi
}

\newcommand{\innermid}{\nonscript\;\delimsize\vert\nonscript\;}
\newcommand{\activatebar}{%
	\begingroup\lccode`\~=`\|
	\lowercase{\endgroup\let~}\innermid 
	\mathcode`|=\string"8000
}

\newcommand\Desc{\textsc{Desc}\xspace}
\newcommand\OPT{\textsc{Opt}\xspace}

\newcommand\ALG{\textsc{Alg}\xspace}

\counterwithin{equation}{section}

\usepackage{eqparbox}

\theoremstyle{plain}

\newtheorem{theorem}{Theorem}[section]

\newtheorem{lemma}[theorem]{Lemma}

\newtheorem{observation}[theorem]{Observation}

\newlength{\continueindent}
\usepackage{etoolbox}
\makeatletter
\newcommand*{\ALG@customparshape}{\parshape 2 \leftmargin \linewidth \dimexpr\ALG@tlm+\continueindent\relax \dimexpr\linewidth+\leftmargin-\ALG@tlm-\continueindent\relax}
\apptocmd{\ALG@beginblock}{\ALG@customparshape}{}{\errmessage{failed to patch}}
\makeatother

\makeatletter
\def\thm@space@setup{%
	\thm@preskip=\parskip \thm@postskip=0pt
}
\makeatother

\usepackage{etoolbox}
\usepackage{tikz}
\usetikzlibrary{tikzmark}
\usetikzlibrary{calc}

\errorcontextlines\maxdimen

\newcommand{\ALGtikzmarkcolor}{black}% customise this, if you want
\newcommand{\ALGtikzmarkextraindent}{4pt}% customise this, if you want
\newcommand{\ALGtikzmarkverticaloffsetstart}{-.5ex}% customise this, if you want
\newcommand{\ALGtikzmarkverticaloffsetend}{-.5ex}% customise this, if you want
\makeatletter
\newcounter{ALG@tikzmark@tempcnta}

\newcommand\ALG@tikzmark@start{%
	\global\let\ALG@tikzmark@last\ALG@tikzmark@starttext%
	\expandafter\edef\csname ALG@tikzmark@\theALG@nested\endcsname{\theALG@tikzmark@tempcnta}%
	\tikzmark{ALG@tikzmark@start@\csname ALG@tikzmark@\theALG@nested\endcsname}%
	\addtocounter{ALG@tikzmark@tempcnta}{1}%
}

\def\ALG@tikzmark@starttext{start}
\newcommand\ALG@tikzmark@end{%
	\ifx\ALG@tikzmark@last\ALG@tikzmark@starttext
	\else
	\tikzmark{ALG@tikzmark@end@\csname ALG@tikzmark@\theALG@nested\endcsname}%
	\tikz[overlay,remember picture] \draw[\ALGtikzmarkcolor] let \p{S}=($(pic cs:ALG@tikzmark@start@\csname ALG@tikzmark@\theALG@nested\endcsname)+(\ALGtikzmarkextraindent,\ALGtikzmarkverticaloffsetstart)$), \p{E}=($(pic cs:ALG@tikzmark@end@\csname ALG@tikzmark@\theALG@nested\endcsname)+(\ALGtikzmarkextraindent,\ALGtikzmarkverticaloffsetend)$) in (\x{S},\y{S})--(\x{S},\y{E});%
	\fi
	\gdef\ALG@tikzmark@last{end}%
}

\apptocmd{\ALG@beginblock}{\ALG@tikzmark@start}{}{\errmessage{failed to patch}}
\pretocmd{\ALG@endblock}{\ALG@tikzmark@end}{}{\errmessage{failed to patch}}
\makeatother
\algblock[with]{With}{EndWith}
\algblockdefx[With]{With}{EndWith}%
[1]{\textbf{with} #1 \textbf{do}}%
{}

\makeatletter
\ifthenelse{\equal{\ALG@noend}{t}}%
{\algtext*{EndWith}}
{}%
\makeatother

%% file: intro.tex
\section{Introduction}

In the Set Cover problem we are given a set system $(U,\mathcal{S})$ (where $U$ is a ground set of size $n$ and $\mathcal{S}$ is a collection of subsets with $|\mathcal{S}| = m$). The goal is to select a minimum size subcollection $\mathcal{S}' \subseteq \mathcal{S}$ such
that the union of the sets in $\mathcal{S}'$ is $U$. Many
algorithms have been discovered for this problem that achieve
an approximation ratio of $1+\ln n$ (see
e.g. \cite{chvatal1979greedy,johnson1974approximation,lovasz1975ratio,williamson2011design}),
and the leading $\ln n$ term is best possible unless $\mathsf{P} = \NP$
\cite{feige1998threshold,Dinur:2014:AAP:2591796.2591884}.

In online Set Cover, the algorithm does not know $U$ initially, nor the contents of each $S \in \mathcal{S}$. An oblivious adversary fixes the set system and arrival order in advance, then reveals each element and the sets containing it, and the algorithm must immediately cover each incoming element using previously selected sets or by irrevocably selecting additional sets. The goal is to pick the smallest number of unique sets over the course of this process. The seminal work of \cite{alon2003online} showed a competitive ratio of $O(\log m \log n)$ for this version.\footnote{This paper considers the harder \emph{unknown-instance model} for online set cover where $n$ means the number of elements requested online (see Chapter 1 of \cite{korman2004use}). The result of \cite{alon2003online} was presented for the \emph{known-instance model} where $n$ means the size of the universe from which requests can be drawn, but extends to the unknown setting as well. This was made explicit in subsequent work \cite{buchbinder2009online}.}

In his Master's thesis \cite{korman2004use}, Simon Korman showed that every polynomial time randomized algorithm requires competitive ratio $\Omega(\log n \log m)$ in the $m= \poly(n)$ regime unless $\NP \subseteq \BPP$. We show that the same bound holds unconditionally up to $\log \log n$ terms even for unbounded computation time algorithms, and even for a much wider regime of $m$ and $n$.

\begin{theorem}\label{thm:general-parameters}
	There is a universal constant $C>0$ such that, for every sufficiently
	large integer $n$, every integer $m$ with
	$(\log n)^2\le m\le2^n$, and every randomized algorithm $\ALG$,
	there is an online set cover instance $\mathcal I$ with $n$ elements
	and $m$ sets such that
	\[
	\frac{\expect*{\abs{\ALG(\mathcal I)}}}{\OPT(\mathcal I)}
	\ge C\frac{\log m\,\log n}{\log\log m+\log\log n}.
	\]
\end{theorem}

This settles (up to $\log \log$ terms) the randomized competitive ratio of the online set cover problem. This improves upon the recent result of \cite{doron2026tight} showing a $\log m \log^{1/2 - \eps} n$ lower bound, which was for a specific setting of $m$ and $n$.

%% file: overview.tex
\section{Overview of the Construction}

At the heart of the construction is an inner instance $\mathcal{I}_k$ governed by a parameter $k$
with $m^{\text{in}}=\Theta(k^2)$ sets and $n^{\text{in}}=2^{\Theta(k)}$ elements. We show that every randomized algorithm has a
competitive ratio of $\Omega(k)=\Omega(\sqrt {m^{\text{in}}})=\Omega(\log n^{\text{in}})$ on such instances. We describe this construction in \cref{sec:in}.

From here we perform the following amplification. We place an independent copy of the inner instance $\mathcal{I}_k(v)$ at each vertex $v$ of a complete binary tree $T^{\text{out}}$ of height $q$ which we call the outer tree.  Elements of the final set system are the union of the elements of all these copies, specified as $(e^{\text{in}}, v^{\text{out}})$ where $e^{\text{in}} \in [n^\text{in}]$ is an element in the inner instance and $v^{\text{out}} \in V(T^{\text{out}})$ is a vertex in the outer tree. Sets in the final system are specified as $(S^{\text{in}}_1, \ldots, S^{\text{in}}_q, \ell^{\text{out}})$, where each $S^{\text{in}}_i \in [m^\text{in}]$ is a set from the inner instance and $\ell^{\text{out}} \in \text{Leaves}(T^{\text{out}})$ is a leaf of the outer tree. In other words, each set is defined by a choice of root to leaf path in the outer tree, and a choice of set for every instance along that path.\footnote{This construction resembles the product construction from \cite{korman2004use}, with some important differences. Firstly, in order to control the instance parameters, \cite{korman2004use} requires that each set in the final system be specified by a root to leaf path and a \emph{single} set from the inner instance, which we interpret as the $q$ copies of the single set along the path. Secondly, \cite{korman2004use} uses as the inner instance hard to approximate set systems obtained via the PCP theorem; this is the source of their assumption that the algorithm run in polynomial time.}

The adversary picks a random root to leaf path $v^{\text{out}}_1, \ldots, v^{\text{out}}_q$ and reveals in order $\mathcal{I}_k({v^{\text{out}}_1}), \ldots,  \mathcal{I}_k({v^{\text{out}}_q})$. We show that every randomized algorithm will incur a competitive ratio of $\Omega(qk)$ on this final sequence. Intuitively, this is because the algorithm guesses the incorrect descendent of every vertex in the outer instance half of the time, and within each inner instance it loses an approximation factor of $k$.

It is instructive to understand parameters in the $m = \poly(n)$ regime. Since the overall set system has $n^{\text{fin}} = qn^{\text{in}} = q 2^{\Theta(k)}$ elements and $m^{\text{fin}} = 2^{q-1} (m^{\text{in}})^q$, setting $q = \Theta(k / \log k)$ yields $n^{\text{fin}} = 2^{\Theta(k)}$ and $m^{\text{fin}}=2^{\Theta(k)}$ and approximation ratio $\Omega(k^2 / \log k) = \Omega(\log^2 n^{\text{fin}} / \log \log n^{\text{fin}})$.

All logarithms in this paper are base 2.

%% file: inner.tex
\section{The Inner Instance}

\label{sec:in}

Our goal in this section is to show the following lemma.

\begin{lemma}
	\label{lem:inner}
	For every integer $k>1$, there is a distribution $\mathcal{I}_k$ over online set cover instances with $m^{\text{in}} = 512 k^2$ and $n^{\text{in}} \leq k4^k$ such that $|\OPT| \leq 2k$ with probability $1$, and any deterministic algorithm $\ALG$ has  $\expect{\abs{\ALG}} \geq k(k+1)/4$. 
\end{lemma}

\subsection{The Construction}

The instance is specified by a complete binary tree $T = (V,E)$ of height $2k$ and a coloring of the vertices $\chi: V \rightarrow [m^{\text{in}}]$. The sets of the set system $\mathcal{S}$ correspond to colors, and each element $e_{H(v)}$ of the set system is specified by a subtree $H(v) \subseteq T$ rooted at a vertex $v$. For convenience, we use the notation $\chi(H(v)):= \{ \chi(u) \mid u \in H(v)\}$. 

We define set-element adjacency as follows. Set $S \in [m^{\text{in}}]$ contains element $e_{H(v)}$ if there is a vertex $u$ on the boundary $\partial(H(v))$ such that $\chi(u) = S$, where
\[\partial(H(v)) := \{v \notin H \mid \text{parent}(v) \in H\} \cup (H \cap \text{leaves(T)}), \qquad \partial(\emptyset(v)) := \{v\}.\]
(Note the special casing of the empty subtree rooted at $v$.)

\input{inner_fig.tex}

Finally, given a fixed tree $T$ and coloring $\chi$, we can describe the online input sequence. The adversary picks a random root to leaf path $\pi = (v_1, v_2, \ldots, v_{2k})$ and reveals in order from $i = 1,\ldots, 2k$, the batch of elements 
\[B_i := \{e_H \mid \text{root}(H) = v_i, |H| < k\}.\] 
(Within each batch, the adversary reveals elements in any canonical lexicographical order.)

Note that $n^{\text{in}}\le 2k\sum_{s=0}^{k-1}4^s<k4^k$. We will set $m^{\text{in}} = 512k^2$.

\subsection{The Analysis}

Bounding the cost of $\OPT$ is straightforward.

\begin{observation}
	The optimal solution satisfies $|\OPT| \leq 2k$ with probability $1$.
\end{observation}
\begin{proof}
	Buy the collection $\{\chi(v_1),\ldots,\chi(v_{2k})\}$.
	Consider any element $e_{H(v_i)}\in B_i$. If $H(v_i)$ is empty,
	then $v_i\in\partial(H(v_i))$. Otherwise, follow the path from
	$v_i$. If it leaves $H(v_i)$, its first vertex outside has its
	parent in $H(v_i)$, and hence belongs to $\partial(H(v_i))$.
	If it never leaves, the terminal leaf $v_{2k}$ lies in $H(v_i)$
	and therefore also belongs to $\partial(H(v_i))$.
\end{proof}

In what remains, we lower bound the cost of the algorithm.

\begin{lemma}\label{lem:fixed-coloring}
	Fix the coloring $\chi: V \rightarrow [m]$ for a tree of height $2k$. Then every deterministic algorithm satisfies
	\[
	\expect{\abs{\ALG}}\ge\frac{k(k+1)}{2}\, \min_{\substack{H(v_1)\subseteq T\\1\le |H(v_1)|\le k^2}}
	\frac{|\chi(H(v_1))|}{|H(v_1)|},
	\]
	where the minimum is over nonempty subtrees rooted at the global
	root $v_1$, and the expectation is over the random choice of root to leaf path $\pi$.
\end{lemma}

\begin{proof}
	 Let $\ALG(i)$ be the collection of sets bought by the algorithm immediately after batch $B_i$, and let $\Desc(i)$ be the descendents of $v_i$ including $v_i$. 
	We will lower bound $|\ALG(k)|$, i.e. the cost after only the first $k$ batches, which suffices to prove the lemma.
	
	For analysis, conditioned on $v_1, \ldots, v_i$, we define a subtree $K_i$ as follows. Define the marked vertices $\textsc{Mark}(i) = \{v \in V(T) \mid \chi(v) \in \ALG(i)\}$ to be those whose colors have been bought. Since $e_{\emptyset(v_i)} \in B_i$ forces $v_i$ to be marked, there is a nonempty connected component $K'_i$ of $\textsc{Mark}(i) \cap \Desc(i)$. The component $K'_i$ must have size at least $k$, otherwise $K'_i$ can contain no leaves (because $v_i$ is at depth at most $k$), which would mean that $\partial(K'_i)$ is unmarked, which in turn would mean that $e_{K'_i} \in B_i$ is uncovered by $\ALG(i)$. Let $K_i \subseteq K'_i$ be an arbitrary connected subtree of size exactly $k$ rooted at $v_i$. 
	
	Given the choices of $K_i$, define $K := \bigcup_{i=1}^k K_i.$
	Observe that:
	\begin{enumerate}
		\item $|K| \leq k^2$,
		\item $K$ contains the path $(v_1, \ldots, v_k)$,
		\item $\chi(K) \subseteq \ALG(k)$ (because each $K_i \subseteq \textsc{Mark}(i)$).
	\end{enumerate}
	
	Since $v_{i+1}$ is chosen uniformly at random from the two children of $v_i$, we have that
	\[
	\expect{\abs{K_i\setminus\Desc(i+1)}\mid v_1,\ldots,v_i}
	=1+\frac{k-1}{2}=\frac{k+1}{2}.
	\]
	Consequently, since $K_i\setminus\Desc(i+1)$ are pairwise disjoint subsets of $K$,
	\[
	\expect{\abs{K}}
	\ge \sum_{i=1}^k\expect{\abs{K_i\setminus\Desc(i+1)}}
	=\frac{k(k+1)}{2}.
	\]
	Since $1\le |K|\le k^2$ and
	$\chi(K)\subseteq\ALG(k)$, we obtain
	\[
	\expect{\abs{\ALG(k)}}
	\ge \expect{\abs{\chi(K)}}
	\ge \expect{\abs{K}}\,
	\min_{\substack{H(v_1)\subseteq T\\1\le |H(v_1)|\le k^2}}
	\frac{|\chi(H(v_1))|}{|H(v_1)|}.  \qedhere
	\]
\end{proof}

To prove \cref{lem:inner}, it suffices to show that for every sufficiently large $k$, there exists a coloring of a complete binary tree of height $2k$ such that every nonempty subtree $H(v_1)$ of size at most $k^2$ rooted at $v_1$ has at least $\Omega(|H(v_1)|)$ many unique colors.

\begin{lemma}\label{lem:color}
	For every positive integer $k$, a complete binary tree $T = (V,E)$ of height $2k$ can be colored with $m=512k^2$ colors so that every subtree $H$ rooted at $v_1$
	with $|H|\le k^2$, satisfies $|\chi(H)|\ge |H|/2$, where  $\chi(H)=\{\chi(v):v\in H\}$.
\end{lemma}

\begin{proof}
	Color every vertex independently and uniformly from $[m]$.
	Fix a subtree $H$ with $1\le |H|=s\le k^2$. Exposing its colors
	in a fixed order, each color repeats an earlier color with
	conditional probability at most $s/m\le1/512$. Hence the number
	$R=s-|\chi(H)|$ of repetitions is stochastically dominated by
	$X\sim\operatorname{Bin}(s,1/512)$.
	Chernoff bounds\footnote{For binomial $X$ and $a>\expect{X}$,
		$\Pr[X>a]\le(e\expect{X}/a)^a$; see \cite[Theorem~4.1, p.~68]{motwani1995randomized}.}
	give
	\[
	\Pr[|\chi(H)|<s/2]
	=\Pr[R>s/2]
	\le\Pr[X>s/2]
	\le(e/256)^{s/2}\le8^{-s}.
	\]
	There are at most $C_s \leq 4^s$ subtrees of size $s$ rooted at $v_1$, where $C_s$ is the $s$-th Catalan number. Summing over these subtrees,
	the probability of any failure is at most
	\[
	\sum_{s=1}^{k^2}4^s8^{-s}
	=\sum_{s=1}^{k^2}2^{-s}=1-2^{-k^2}<1.
	\]
	Hence there exists at least one coloring that satisfies the desired inequality simultaneously for every subtree $H$.
\end{proof}

\cref{lem:inner} follows by combining \cref{lem:fixed-coloring} with \cref{lem:color}. 

%% file: inner_fig.tex
\begin{figure}[!ht]
\centering
\begin{tikzpicture}[x=1cm,y=1cm,font=\small,
  color vertex/.style={draw=ink!60,rounded corners=1pt,
    minimum width=4.6mm,minimum height=4.6mm,inner sep=0pt,
    line width=.45pt},
  boundary color/.style={color vertex,double=white,
    double distance=1.1pt,line width=.55pt}]

  % H(v) includes both children of v and one child of its left child.
  \draw[draw=ink!45,rounded corners=5pt,
    pattern=north east lines,pattern color=ink!16]
    (-2.45,-.64)--(-1.77,-.64)--(-.02,-1.91)--(-.02,-2.61)
    --(-1.74,-2.61)--(-2.52,-3.86)--(-3.27,-3.86)
    --(-4.13,-2.57)--(-4.13,-1.98)--cycle;
  \node at (-2.10,-2.35) {$H(v)$};

  % The global root and the other branch of T give the local context.
  \node[color vertex,fill=ochre] (r) at (0,0) {};
  \node[above=2mm of r] {root of $T$};
  \node[color vertex,fill=pathcolor] (s) at (1.80,-1.05) {};
  \node[color vertex,fill=teal] (sl) at (1.05,-2.25) {};
  \node[color vertex,fill=ochre] (sr) at (2.55,-2.25) {};

  % The four vertices of H(v).
  \node[color vertex,fill=ochre] (v) at (-2.10,-1.05) {};
  \node[right=2mm of v] {$v$};
  \node[color vertex,fill=teal] (a) at (-3.70,-2.25) {};
  \node[color vertex,fill=pathcolor] (b) at (-.50,-2.25) {};
  \node[color vertex,fill=ochre] (c) at (-2.90,-3.45) {};

  % The boundary has five vertices, at two different depths.
  \node[boundary color,fill=pathcolor] (al) at (-4.50,-3.45) {};
  \node[boundary color,fill=teal] (bl) at (-1.30,-3.45) {};
  \node[boundary color,fill=pathcolor] (br) at (.30,-3.45) {};
  \node[boundary color,fill=teal] (cl) at (-3.30,-4.65) {};
  \node[boundary color,fill=pathcolor] (cr) at (-2.50,-4.65) {};

  \foreach \u/\w in {r/v,r/s,s/sl,s/sr,v/a,v/b,a/al,a/c,b/bl,b/br,c/cl,c/cr}
    \draw[treeedge] (\u.south)--(\w.north);

  % Each displayed branch continues in the complete binary tree.
  \foreach \u in {sl,sr,al,bl,br,cl,cr} {
    \draw[treeedge] (\u.south)--++(-.20,-.30);
    \draw[treeedge] (\u.south)--++(.20,-.30);
    \node[text=ink!55] at ($(\u.south)+(0,-.60)$) {$\vdots$};
  }
\end{tikzpicture}
\caption{Illustration of a subtree $H(v)$ rooted at a vertex $v$. The element $e_{H(v)}$ corresponding to $H(v)$
belongs exactly to the purple and teal sets because vertices of these colors appear on $\partial(H(v))$.}
\label{fig:inner-instance}
\end{figure}
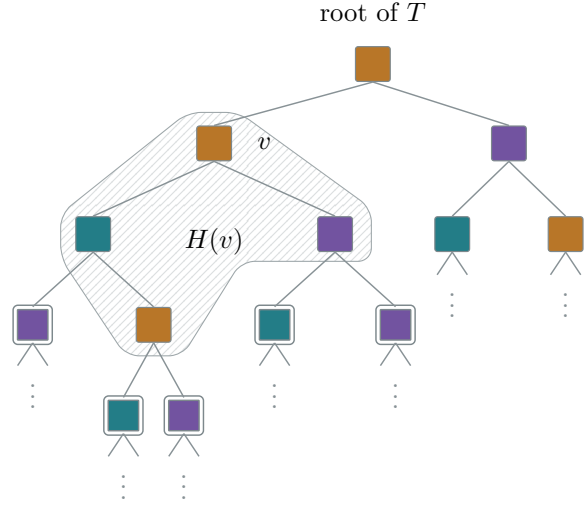

%% file: amplify.tex
\section{Binary Tree Amplification}\label{sec:tree-product}

We use $n^{\text{in}},m^{\text{in}}$ for the inner numbers of elements and sets, and
$n^{\text{fin}},m^{\text{fin}}$ for the final numbers of elements and
sets. The outer tree $T^{\text{out}}$ has $q$ levels, including the
root, so each root to leaf path has $q$ vertices.

\subsection{The Construction}

\begin{lemma}\label{lem:tree-product}
	Suppose there is a distribution $\mathcal I$ over online set cover instances with $m$ sets and $n$ elements such that
	$|\OPT|\le L$ and any deterministic algorithm $\ALG$ has expected cost at least $\beta$ on the inner instance. Then there is a distribution  $\mathcal I^\text{fin}$ over instances with exactly $n^{\text{fin}} = qn$ elements and $m^{\text{fin}} = 2^{q-1}m^q$
	 sets such that $|\OPT|\le L$ and $\expect{\abs{\ALG}}\ge q\beta/2$.
\end{lemma}

\begin{proof}
	Let $T^{\text{out}}$ be a complete binary tree with $q$ levels.
	Independently sample an inner request sequence at each vertex
	$v^{\text{out}} \in T^{\text{out}}$, and denote this copy by $\mathcal{I}(v^{\text{out}})$.
	Index its $n^{\text{in}}$ request occurrences by $e^{\text{in}}\in[n^{\text{in}}]$, in arrival
	order. 
	
	Elements of the \emph{final} set system are the union of the elements of all these copies, specified as $(e^{\text{in}}, v^{\text{out}})$ where $e^{\text{in}} \in [n^{\text{in}}]$ is an element in the inner instance and $v^{\text{out}} \in V(T^{\text{out}})$ is a vertex in the outer tree. A set $S$ in the final instance is specified by a tuple
	\[
	S=(S^{\text{in}}_1,\ldots,S^{\text{in}}_q,\ell^{\text{out}}),
	\qquad S^{\text{in}}_i\in[m^\text{in}],\quad
	\ell^{\text{out}}\in\text{Leaves}(T^{\text{out}}).
	\]
	This set is defined to cover precisely the union over $i \in [q]$ of the elements of $S^{\text{in}}_i$ in the copy of $\mathcal{I}$ at level $i$ of the path
	$P(\ell^{\text{out}})$.
	
\input{product.tex}

	The final request sequence is the following. Independently of the inner samples, choose a uniform root to leaf path
	$\overline P = (\overline v^{\text{out}}_1,\ldots,\overline v^{\text{out}}_q)$.
	In round $i$, reveal the $n$ requests of
	$\mathcal{I}(\overline v^{\text{out}}_i)$ in order, with incidences
	\begin{align*}
		(e^{\text{in}},\overline v^{\text{out}}_i)\in (S^{\text{in}}_1,\ldots,S^{\text{in}}_q,\ell^{\text{out}})
		\quad\Longleftrightarrow\quad
		\overline v^{\text{out}}_i\in P(\ell^{\text{out}})
		\ \text{ and }\ e^{\text{in}}\in S^{\text{in}}_i.
	\end{align*}
	
	The final instance $\mathcal{I}^\text{fin}$ has $m^{\text{fin}}=2^{q-1}(m^\text{in})^q$ final sets, and $n^\text{fin} = q n^\text{in}$ elements. (We do not include the elements of $\mathcal{I}(v^{\text{out}})$ for any $v^{\text{out}} \not \in \overline P$ because these are not actually requested of the algorithm.)

	\subsection{The Analysis}
	\paragraph{Upper bounding $\OPT$.}

	Let $\OPT_i = \{S_1^\text{in}(i), \ldots, S_L^\text{in}(i)\}$ be the optimal cover for the instance in round $i$. Then $\OPT^\text{fin}$ can pick the sets $(S_j^\text{in}(1), \ldots, S_j^\text{in}(q), \overline v^{\text{out}}_q)$ for $j \in [L]$, and hence $|\OPT| \leq L$. (In words, $\OPT$ always picks the correct path in the outer tree, and each of its final sets selects one of the $\leq L$ inner sets from each of the instances on the correct path.)
	
	\paragraph{Lower bounding $\ALG$.} 
	Let $\mathcal A_i$ be
	the sets $(S^{\text{in}}_1,\ldots,S^{\text{in}}_q,\ell^{\text{out}})$ that the algorithm owns at the end of round $i$ for which
	$\overline v^{\text{out}}_i\in P(\ell^{\text{out}})$ (i.e. sets that meaningfully contribute to covering the instance $\mathcal{I}(\overline v^{\text{out}}_i)$ revealed in round $i$). Let $Z_i$ be the sets bought in round $i$. For convenience, define $\mathcal A_0=\emptyset$. Projected onto the inner instance $\mathcal{I}(\overline v^{\text{out}}_i)$, the sets in $\mathcal{A}_i$ are an online algorithm for this instance. Thus, by assumption on $\mathcal{I}$, we have $\expect{\abs{\mathcal A_i}} \geq \beta$. On the other hand, since the choice of whether $\overline v^{\text{out}}_i$ is the left or right child of $\overline v^{\text{out}}_{i-1}$ is independent of $\mathcal{I}(\overline v^{\text{out}}_1), \ldots, \mathcal{I}(\overline v^{\text{out}}_{i-1})$ and of the prefix of the outer path $\overline v^{\text{out}}_{1}, \ldots, \overline v^{\text{out}}_{i-1}$, we have that every set of $\mathcal A_{i-1}$ is in $\mathcal A_i$ with probability $1/2$. Thus
	\[
	\expect{\abs{\mathcal A_i}}
	\le\tfrac12\expect{\abs{\mathcal A_{i-1}}}+\expect{\abs{Z_i}},
	\]
	This also holds for $i=1$, since $\mathcal A_1\subseteq Z_1$.
	Summing and using the disjointness of the $Z_i$ gives
	\[
	\expect{\abs{\ALG}}
	\ge\sum_{i=1}^q\expect{\abs{Z_i}}
	\ge\sum_{i=1}^q\left(\expect{\abs{\mathcal A_i}}
	-\tfrac12\expect{\abs{\mathcal A_{i-1}}}\right)
	\ge\tfrac12\sum_{i=1}^q\expect{\abs{\mathcal A_i}}
	\ge q\beta/2. \qedhere
	\]
\end{proof}

%% file: product.tex
\begin{figure}[!ht]
\centering
\begin{tikzpicture}[x=1cm,y=1cm,font=\small,
  copy/.style={draw=ink!55,fill=white,rounded corners=3pt,
    minimum width=32mm,minimum height=18.5mm,inner sep=0pt},
  off path/.style={copy,pattern=north east lines,pattern color=ink!22},
  inner set/.style={draw=ink!45,fill=ink!3,rounded corners=1pt,
    minimum width=4.6mm,minimum height=5mm,inner sep=0pt,
    font=\footnotesize},
  selected set/.style={inner set,draw=pathcolor,fill=pathcolor,
    text=white,line width=.8pt}]
  % Display the inner sets along the chosen path. The fourth argument
  % selects a rectangle; the fifth gives the level.
  \newcommand{\productcopy}[5]{%
    \node[copy] (#1) at (#2) {};
    \begin{scope}[shift={(#1.center)}]
      \node at (0,.56) {$\mathcal I(#3)$};
      \foreach \j/\x/\lab in {1/-1.10/1,2/-.48/2,3/.14/3,4/1.12/m} {
        \ifnum\j=#4\relax
          \node[selected set] at (\x,-.06) {$\lab$};
          \node[text=pathcolor] at (\x,-.61) {$S^{\text{in}}_{#5}$};
        \else
          \node[inner set] at (\x,-.06) {$\lab$};
        \fi
      }
      \node[font=\footnotesize] at (.64,-.06) {$\cdots$};
    \end{scope}
  }

  \productcopy{r}{0,0}{v^{\text{out}}_1}{1}{1}
  \productcopy{l}{-3.65,-2.8}{v^{\text{out}}_2}{3}{2}
  \node[off path] (s) at (3.65,-2.8) {};
  \node[off path] (ll) at (-5.475,-5.6) {};
  \productcopy{lr}{-1.825,-5.6}{v^{\text{out}}_3}{2}{3}
  \node[off path] (sl) at (1.825,-5.6) {};
  \node[off path] (sr) at (5.475,-5.6) {};

  \foreach \u/\v in {r/l,r/s,l/ll,l/lr,s/sl,s/sr}
    \draw[treeedge] (\u.south)--(\v.north);
  \draw[draw=pathcolor,line width=1.3pt]
    (r.south)--(l.north) (l.south)--(lr.north);
  \node[below=4mm of lr,text=pathcolor]
    {$\ell^{\text{out}}=v^{\text{out}}_3$};
  \node at (5.1,0) {outer tree $T^{\text{out}}$};
\end{tikzpicture}
\caption{Representation of the final instance, with a single set $S=(1,3,2,\ell^{\text{out}})$ of the outer instance highlighted.}
\label{fig:tree-product}
\end{figure}
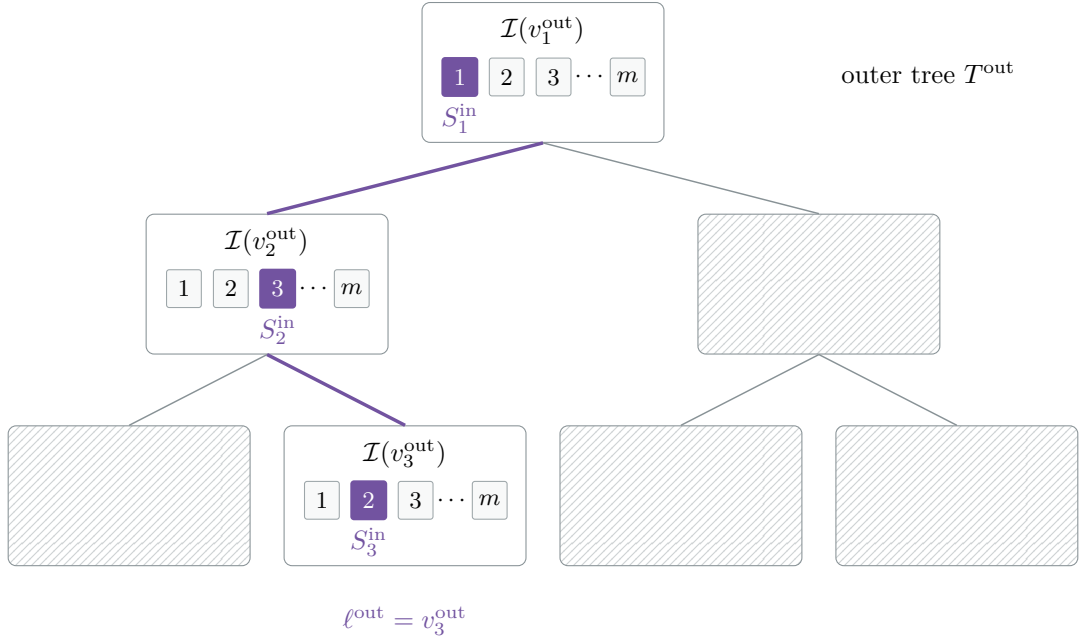

%% file: general_parameters.tex
\section{Proof of \cref{thm:general-parameters}}

We prove the theorem by applying \cref{lem:tree-product} with the instance from \cref{lem:inner}. For a fixed integer $k$, the inner instance $\mathcal{I}_k$ has parameters 
\[
n^\text{in} \leq k4^k, \qquad m^\text{in} = 512 k^2, \qquad |\OPT^\text{in}| \leq 2k, \qquad \expect{\abs{\ALG^\text{in}}} \geq \frac{k(k+1)}{4}.
\]
Hence the final amplified instance has parameters
\[
n^\text{fin} \leq qk4^k , \qquad m^\text{fin} = 2^{10q-1} k^{2q}, \qquad |\OPT^\text{fin}| \leq 2k, \qquad \expect{\abs{\ALG^\text{fin}}} \geq \frac{qk(k+1)}{8}.
\]
By padding, we can use any parameters $n$ and $m$ so long as $n \geq n^\text{fin}$ and $m \geq m^\text{fin}$.

By assumption, $\log^2 n \leq m \leq 2^{n}$. We set the parameters $k$ and $q$ by cases.

\emph{Case 1: $m\le2^{\sqrt n}$.}
Choose
\[
k= \left \lfloor\frac {\log n}{32} \right \rfloor,
\qquad q=\left\lfloor\frac{ \log m} {10+2\log k}\right\rfloor \leq \frac{\sqrt{n}}{2 \log k}.
\]
For sufficiently large $k$ and $q$, this yields
\[
n \geq n^{1/2}  2^{16 k} = n^{1/2} 4^k 2^{14k}  \geq (2q \log k) 4^k 2^{14k} = \frac{2^{14k+1}\cdot \log k }{k} \cdot qk4^k \geq n^\text{fin}.
\]
On the other hand,
\[
m \geq 2^{q (10 + 2\log k)} \geq 2^{10q-1}k^{2q} = m^{\text{fin}}.
\]
The competitive ratio is at least 
\[
\frac{q(k+1)}{16} \geq \frac{\log m \log n}{4096 \log \log n}.
\]

\emph{Case 2: $m>2^{\sqrt n}$.}
Choose
\[
k=2,\qquad q=\left\lfloor\frac{\log m}{32}\right\rfloor.
\]
For sufficiently large $n$, $q\ge(\log m)/64\ge1$, giving
\[
n\ge\log m\ge32q\ge n^\text{fin},
\qquad
m\ge2^{32q}\ge2^{12q-1}=m^\text{fin}.
\]
Since $m>2^{\sqrt n}$, we have $\log\log m>\tfrac12\log n$.
Thus the competitive ratio is at least
\[
\frac{q(k+1)}{16}
\ge\frac{\log m}{512}
\ge\frac{\log m\,\log n}{1024\log\log m}.
\]

In conclusion, for any sufficiently large $\log^2 n \leq m \leq 2^n$, there exists a distribution over instances with $m$ sets and $n$ elements for which any deterministic algorithm has competitive ratio \[ \frac{1}{4096} \cdot \frac{\log m \log n}{ \log \log n + \log \log m}.\] 
The theorem follows by Yao's principle.